%% file: root.tex
\PassOptionsToPackage{compatibility=false}{caption}
\documentclass[letterpaper, 10 pt, journal, twoside]{ieeetran}
\IEEEoverridecommandlockouts

\usepackage{mathtools}
\usepackage{siunitx}
\usepackage{tabularx}
\usepackage[usenames,dvipsnames]{xcolor}
\usepackage{pgfplots}
\usepgfplotslibrary{fillbetween}
\pgfplotsset{compat = newest}
\usetikzlibrary{arrows,positioning,shapes,intersections,patterns,calc,fit,external,decorations,decorations.markings} 
\usepackage{cite}
\usepackage{textcomp}
\usepackage{inputenc}
\usepackage[american]{babel}
\usepackage{csquotes}

\ifCLASSOPTIONcompsoc
  \usepackage[nocompress]{cite}
\else
  \usepackage{cite}
\fi

\usepackage{comment}
\usepackage[listings,light]{solarized}
\input{tikz.tex}
\usepackage{aircraftshapes}
\usepackage{url}
\usepackage{hyperref}
\hypersetup{colorlinks=false,linkcolor=blue,urlcolor=blue}
\usepackage{graphicx}
\usepackage{subcaption}

\renewcommand{\figurename}{Fig.}
\addto\captionsenglish{\renewcommand{\figurename}{Fig.}}
\addto\captionsamerican{\renewcommand{\figurename}{Fig.}}
\usepackage{multicol}
\usepackage{xspace}
\usepackage{mathtools}
\usepackage{booktabs}
\usepackage{multirow}
\usepackage{tabularx}
\newcolumntype{C}{>{\centering\arraybackslash}X}
\newcolumntype{L}{>{\raggedright\arraybackslash}X}
\usepackage{siunitx}
\usepackage{listings}
\usepackage{float}
\usepackage{bm}         
\usepackage[normalem]{ulem}

\usepackage{amsmath,amssymb,amsfonts,amsthm}
\newtheorem{lemma}{Lemma}

\usepackage{bm}

\newcommand{\R}{\mathbb{R}}

\newcommand{\N}{\mathbb{N}}

\newcommand{\D}{\mathcal{D}}
\newcommand{\X}{\mathcal{X}}
\newcommand{\Verts}[1]{{\left\Vert #1 \right\Vert}}
\DeclareMathOperator{\diag}{diag}
\DeclareMathOperator{\var}{var}
\DeclareMathOperator{\mean}{\mu}
\DeclareMathOperator{\Var}{\Sigma}
\DeclareMathOperator{\Mean}{\bm\mu}

\DeclareMathOperator{\Prob}{P}
\newtheorem{assum}{Assumption}
\newtheorem{prop}{Proposition}
\newtheorem{cor}{Corollary}
\newtheorem{rem}{Remark}
\usepackage[noabbrev]{cleveref} 
\crefname{rem}{Remark}{Remarks}
\crefname{exam}{Example}{Examples}
\crefname{assum}{Assumption}{Assumptions}
\crefname{prop}{Proposition}{Propositions}
\crefname{propy}{Property}{Properties}
\crefname{cor}{Corollary}{Corollaries}
\crefname{lem}{Lemma}{Lemmas}
\crefname{section}{Section}{Sections}
\crefname{thm}{Theorem}{Theorems}
\crefname{defn}{Definition}{Definitions}
\crefname{figure}{Fig.}{Fig.}
\Crefname{figure}{Figure}{Figures}
\crefname{equation}{}{}

\newcommand{\x}{\bm x}

\newcommand{\f}{\bm{f}}

\newcommand{\y}{\bm{y}}

\usepackage{mathtools}
\usepackage{siunitx}
\usepackage{tabularx}
\usepackage[usenames,dvipsnames]{xcolor}
\usepackage{pgfplots}
\usepgfplotslibrary{fillbetween}
\pgfplotsset{compat = newest}
\usetikzlibrary{arrows,positioning,shapes,intersections,patterns,calc,fit,external,decorations,decorations.markings} 
\usepackage{booktabs}
\allowdisplaybreaks
\begin{document}

\title{Learning-Based Fault-Tolerant Control for an Hexarotor with Model Uncertainty}

\author{Leonardo J. Colombo$^{1}$,
    Manuela Gamonal Fern\'andez$^{2}$,
   and Juan I. Giribet$^{3}$.

\thanks{J. Giribet was partially supported by NVIDIA Applied Research Program Award 2021, PICT-2019-2371 and PICT-2019-0373 projects from Agencia Nacional de Investigaciones Cient\'ificas y Tecnol\'ogicas, and UBACyT-0421BA project from the Universidad de Buenos Aires (UBA), Argentina. M. Gamonal was partially supported under ``Severo Ochoa Program for Centres of Excellence'' (CEX2019-000904-S). L. Colombo and M. Gamonal acknowledge financial support from the Spanish Ministry of Science
and Innovation, under grants PID2019-106715GB-C21. This work was supported by a 2020 Leonardo Grant for Researchers and Cultural Creators, BBVA Foundation. The BBVA Foundation accepts no responsibility
for the opinions, statements and contents included in the project and/or the results thereof, which are entirely the responsibility of the authors.}
\thanks{$^{1}$Leonardo J. Colombo is with the Centre for Automation and Robotics
(CSIC-UPM), Ctra. M300 Campo Real, Km 0,200, Arganda
del Rey - 28500 Madrid, Spain.
        {\tt\footnotesize leonardo.colombo@car.upm-csic.es}}
        \thanks{$^{2}$Manuela Gamonal is with the Institute of Mathematical Sciences, ICMAT (CSIC-UAM-UCM-UC3M), Spain. 
        {\tt\footnotesize manuela.gamonal@icmat.es}}
        \thanks{$^{3}$Juan I. Giribet is with Universidad de San Andr\'es (UdeSA) and CONICET, Argentina.
        {\tt\footnotesize jgiribet@conicet.gov.ar}}
%
}

\markboth{IEEE Robotics and Automation Letters. Preprint Version.} 
{Colombo \MakeLowercase{\textit{et al.}}: Learning-based fault-tolerant control}
\maketitle

\begin{abstract}

In this paper we present a learning-based tracking controller based on Gaussian processes (GP) for a fault-tolerant hexarotor in a recovery maneuver. In particular, to estimate certain uncertainties that appear in a hexacopter vehicle with the ability to reconfigure its rotors to compensate for failures. The rotors reconfiguration introduces disturbances that make the dynamic model of the vehicle differ from the nominal model. The control algorithm is designed to learn and compensate the amount of modeling uncertainties after a failure in the control allocation reconfiguration by using GP as a learning-based model for the predictions. In particular the presented approach guarantees a probabilistic bounded tracking error with high probability.  The performance of the learning-based fault-tolerant controller is evaluated through experimental tests with an hexarotor UAV.
\end{abstract}


\section{Introduction}\label{sec:intro}

Recent decades have seen an exponential growth in the development and use of autonomous aerial vehicles, which are capable of transporting passengers and cargo, performing remote monitoring or dangerous tasks with a significant reduction in operating costs due to the reduced need for human operators, both on-board and on the ground. These vehicles have applications on diverse domains as passenger transport and logistics, infrastructure monitoring, agriculture, early response to natural disasters and emergencies, communications relay, and provision of internet services, among many others. Globally, autonomous aerial vehicles represent a market of approximately USD 4 billion and it is projected that by 2030 will reach USD 24 billion \cite{AAM1}.

\begin{figure}[t]
    \centering
    \includegraphics[width = 0.8\columnwidth]{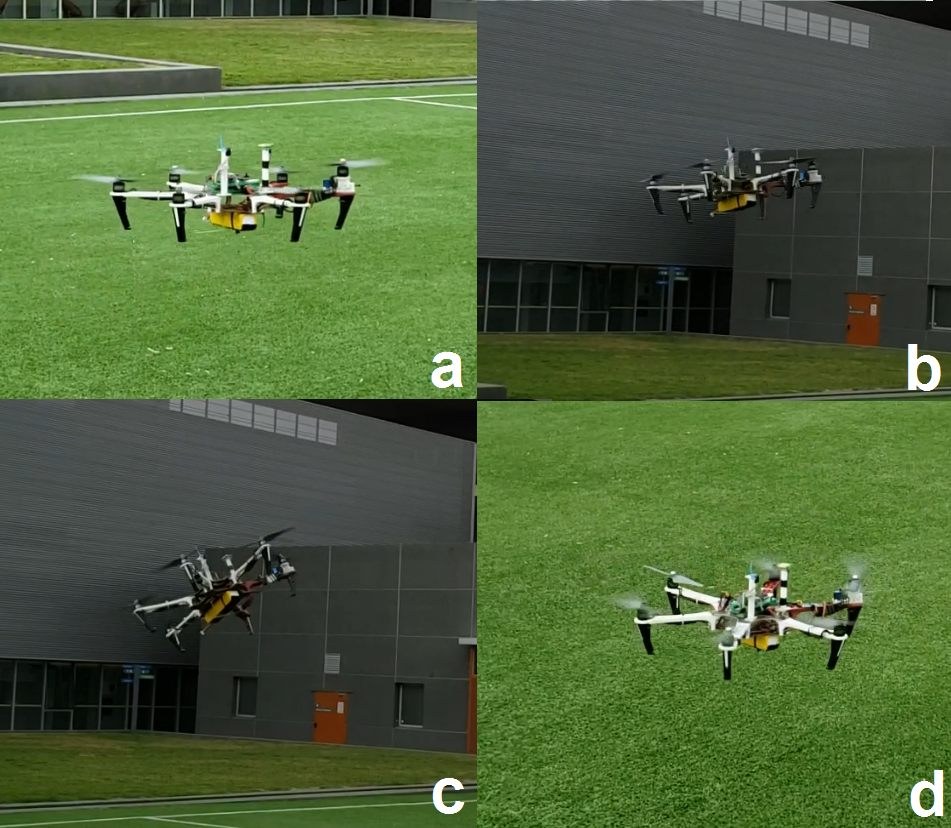}
    \caption{Behavior of the fault-tolerant hexarotor during a failure, detection and recovery maneuver. (a) Vehicle in nominal conditions. (b) One motor stops. (c) The failure is detected and the control allocation is reconfigured. (d) The failure is compensated and the vehicle is stabilized.}
    \label{fig:vuelo_outdoor}
\end{figure}

In the last decade, due to the progress of smaller, more powerful, and cheaper sensors and computers, and the convergence of distributed electric propulsion and storage technologies, in particular, multirotor type vehicles have become one of the better alternatives due to their maneuverability, being able to operate in small spaces and without requiring large dedicated infrastructures, thus enabling, among other things, operation in urban environments. The market potential for multirotor has driven the creation of a phenomenal number of technology-based companies dedicated to the development of these technologies, reaching a total of 10 billion in investment in the last 5 years, of which almost 5 billion correspond to the last year only \cite{AAMFund1}.

The autonomous operation of this type of aircraft requires addressing a set of problems. This is not only the case during a nominal flight, where all systems on-board and on the ground are operating within nominal parameters and the challenge is to flight in a confined space without collision, but it is significantly aggravated in emergency cases, where the vehicle must execute a forced emergency landing in conditions of reduced maneuverability and uncontrolled spaces with the possibility of damage to third parties. This remains to this day a critical point in the certification of UAVs \cite{NASAAAM}.


In \cite{giribet2016analysis}, the capability of compensating for a rotor failure without losing the ability to exert torques in all directions, and therefore keeping full attitude control in case of failure, is studied. For this, at least six rotors are needed, and was shown that an hexacopter with tilted rotors is fault tolerant.
While the system proved to work correctly, there was a direction in which, the achievable torque is noticeably reduced.
In \cite{Michieletto2017}, a detailed analysis is made for the optimal orientations of the rotors in a hexacopter, to achieve full tolerant attitude control. Still, the maximum torque achievable in some directions may be too small, with the consequent degradation of vehicle maneuverability. 
To overcome this limitation, in \cite{TMECH2020} several hexacopter structures and their fault-tolerant capabilities are analyzed. It is shown that a possible solution is to convert these structures into re-configurable ones, to significantly improve the maneuverability in case of a total failure in one rotor. 
By adding a mechanism to tilt one rotor sideway in case of a failure, a standard hexarotor vehicle can be converted into a robust fault-tolerant one. Furthermore, it was shown that it is enough to install this tilting mechanism in only two rotors, reducing the necessity of mechanical parts, improving the reliability of the vehicle. 
In previous mentioned works, the fault tolerant problem was studied as a control allocation problem, i.e., given a desired torque $\bm{\tau}_{cmd}\in\mathbb{R}^3$ and vertical force $f_z\in\mathbb{R}$ computed by the control algorithm, the problem is to find the PWM signals $u=[u_1,...,u_6]$ ($0\leq u_i\leq 100\%$) commanded to each rotor, in order to achieve the control torque. For this, the control allocation matrix $A\in\mathbb{R}^{4\times 6}$ is computed, which depends on the position and orientation of each rotor, usually with the Moore-Penrose pseudoinverse $u=A^\dagger \begin{bmatrix}\bm{\tau}_{cmd} \\ f_z\end{bmatrix}$.

When a failure occurs, and the rotors are configured in order to compensate the failure, a new allocation matrix $A_i\in\mathbb{R}^{4\times 6}$ is computed (which depends on the rotor that failed). The control algorithm remains the same in nominal and under failure conditions, it only changes how the force is distributed among the remaining rotors, in order to achieve the desired control torque $\bm{\tau}_{cmd}$ and force $f_z$. 

At first glance, since the control algorithm is not modified, the vehicle performance should be the same. However, this is not the case in practice because several factors affect the control performance after a failure. For instance, in nominal conditions the vehicle is designed considering that the PWM signals are between 40\% and 60\% in normal flight conditions. Then the hypothesis that there is a lineal relation between PWM and commanded torque and forces is valid. When a failure occurs the PWM must be redistributed and some rotors work between 40\% and 60\% meanwhile others works at 60\% and 80\%, approximately. Then the non-linearity relation between the PWM and commanded force and torque is more noticeable. This could be compensated with a correct characterization of each rotor, but this is not an easy task. Usually this effect is considered as a disturbance for the control algorithm and mitigated when the control loop is closed, scarifying performance because there is always a trade-off between performance and robustness. 
But even with a correct calibration of the rotors, this is not the only effect that affects the performance. In nominal conditions in hovering, the roll and pitch angle moves around zero. But when a failure occurs, in order to maintain hovering, these angles are moved away from zero because the rotor forces are not balanced. This also affects the aerodynamics forces acting on the vehicle, which are not easily characterized. Furthermore, recently has been shown in \cite{Abbaraju2021} that the aerodynamic effects caused due to tilt angled propeller configurations has impact on the vehicle performance, and in particular characterize experimentally some effects that were not particularly considered before as the  blade flapping effect for cant angled propellers. The authors proposed an aerodynamic model for these effects, and could be compensated but wind tunnel experiments must be carried out, which is not an easy task. Furthermore, it is not easy to isolate these effects and some others which we are not aware. 
All these model disturbances can be noticed experimentally. Even when the control algorithm remains the same after a failure, the vehicle performance is affected when switching from nominal to failure conditions. 

In this situation learning-based control strategies seem to be a valid strategy in order to compensate this model errors. In this work we explore a learning-based control strategy, where the learning is based on Gaussian processes (GP) models~\cite{rasmussen2006gaussian}. 

Recently, GP models has been increasingly used for modeling dynamical system due to many beneficial properties such as the bias-variance trade-off and the strong connection to Bayesian statistics \cite{beckers2021introduction}. In contrast to many other techniques, GP models provide not only a prediction but also a measure for the uncertainty of the model. This powerful property makes them very attractive for many applications in control, e.g., model predictive control~\cite{hewing2019cautious}, sliding mode control~\cite{lima2020sliding}, tracking of mechanical systems~\cite{beckers2019stable}, and backstepping control for underactuated vehicles~\cite{beckers2021online}, as the uncertainty measure allows to provide performance and safety guarantees. 
The purpose of this article is to employ the power of learning-based approaches based on GP models to learn the uncertainties in the model after a failure is detected guaranteeing the probabilistic boundedness of the tracking error to the reconfigured attitude and positions with high probability, together with an experimental tests for the validation of learning-based the controller in real practice. This allows to improve the model and, thus, mitigate the model uncertainties during runtime.

The remainder of this paper is structured as follows: after the problem setting in~\cref{sec:ps}, the learning-based modeling with GP and the tracking controller are introduced in~\cref{sec:lbc}. In particular, we provide a probabilistic model error bound for the unknown dynamics and apply it in to design a data-driven tracking control law for position and attitude of the UAV, based on a feedback control, and provide safety guarantees by means of a probabilistic ultimate bound of the tracking error. Finally, the performance  of the learning-based fault-tolerant controller is evaluated through experimental tests with an hexarotor UAV in~\cref{sec:exp}.

\section{Problem Setting}\label{sec:ps}
We assume a single rigid body on $SE(3)$\footnote{Vectors are denoted with bold characters and matrices with capital letters. The term~$A_{i,:}$ denotes the $i$-th row of the matrix~$A$. The expression~$\mathcal{N}(\mu,\Sigma)$ describes a normal distribution with mean~$\mu$ and covariance~$\Sigma$. The probability function is denoted by $\Prob$. The set $\R_{>0}$ denotes the set of positive real numbers. $||\cdot||$ denotes the Euclidean norm. $SO(3)$ and $SE(3)$ denote the special orthogonal and special euclidean Lie groups, respectively.} with position $\bm{p}\in\R^3$ and orientation matrix $R\in SO(3)$. The body-fixed angular velocity is denoted by $\bm{\omega}\in\R^3$ and the linear velocity by $\bm{v}\in\mathbb{R}^{3}$. The vehicle has mass matrix $m>0$ and rotational inertia tensor $J\in\R^{3\times 3}$, symmetric and positive definite. The state space of the vehicle is $Q=SE(3)\times\R^6$ with $\bm{q}=((R,\bm{p}),(\bm{\omega},\bm{v}))\in Q$ denoting the whole state of the system. The vehicle is actuated with control input vectors $\bm{u}_1\in\R^3$ and $\bm{u}_2\in\R^3$, representing the $6$D generalized actuation force acting on the system. 

We can model the system with the following set of differential equations representing the kinematics of the rigid body and its uncertain (i.e., partially unknown) dynamics
\begin{align}
	\begin{split}
	\dot{R}&=RS(\bm{\omega}),\,\,\dot{p}=R\bm{v},\\
	J\dot{\bm{\omega}}&=-\bm{\omega}\times J \bm{\omega}+\bm{u}_1+\f_{\omega}(\bm{q})\big.,\\
	m \dot{\bm{v}}&=-\bm{\omega}\times m\bm{v}+mgRe_3+\bm{u}_2+\f
	_{v}(\bm{q}),
	\end{split}\label{for:system}
\end{align}
where the operator $S\colon\R^3\to \mathfrak{so}(3)$ is given by
$\displaystyle{
	S(\bm{\omega})=\begin{bmatrix}
	0 & -\omega_3 & \omega_2\\
	\omega_3 & 0 & -\omega_1\\
	-\omega_2 & \omega_1 & 0
	\end{bmatrix}}$, and $\mathfrak{so}(3)$ is the Lie algebra of $SO(3)$ determined by the set of $(3\times 3)$ skew-symmetric matrices. The functions $\f_v\colon Q\to\R^3$ and $\f_\omega\colon Q\to\R^3$ are state-depended unknown dynamics. It is assumed that the full state $\bm{q}$ can be measured. The general objective is to track a desired trajectory described by the functions $(R_d,\bm{p}_d)\colon [0,T]\to SE(3)$. Even though in this formulation we consider a fully actuated system, for some under-actuated systems a virtual control input can be defined to transform the system in a suitable form (see \cite{isidori1985nonlinear}).

In preparation for the learning and control step, we transform the system dynamics~\eqref{for:system} in an equivalent form. 

Let us define the vector $\bm{\xi}=[\bm{\omega}^{T},\bm{v}^{T}]^{T}\in\mathbb{R}^6$. Then, the vector space isomorphism $\breve{S}:\mathbb{R}^6\to\mathfrak{se}(3)$ is given by 
$$\breve{S}(\bm{\xi})=\breve{S}(\bm{\omega},\bm{v})=\begin{bmatrix}
	S(\bm{\omega}) &  \bm{v}\\
	0_{1\times3} & 0 
	\end{bmatrix}\in\mathfrak{se}(3)$$ where $\mathfrak{se}(3)$ is the Lie algebra of $SE(3)$ determined by the matrices
$\displaystyle{\left(
     \begin{array}{cc}
      A & b\\
       0 & 0 \\
     \end{array}
   \right)}$, with $A\in \mathfrak{so}(3)$ and $a\in \mathbb{R}^3$. By defining 
\begin{align}
	D(\bm{\omega},\bm{v})=\begin{bmatrix}
	-S(\bm{\omega}) & -S(\bm{v})\\ 0_{3\times 3} & -S(\bm{\omega})
	\end{bmatrix},\, \bm{G}=\begin{bmatrix}
	R & \bm{p}\\ 0_{3\times 1} & 1
	\end{bmatrix}\in SE(3),
\end{align}$\mathbb{I}=\hbox{diag}\{J,M\}$, we can write~\eqref{for:system} in matrix form as 
\begin{equation}\label{for:system1}\begin{cases}
	\dot{\bm{G}}&=\bm{G}\breve{S}(\bm{\xi}) \\
				\mathbb{I}\dot{\bm{\xi}}&=D(\bm{\omega},\bm{v})\mathbb{I}\bm{\xi}+\bm{u}+\bm{f}(\bm{s})
	\end{cases}\end{equation} ${\bm{f}}=({\bm{f}}_v,{\bm{f}}_w):Q\to\mathbb{R}^6$, $\bm{q}\in Q$. So,~\eqref{for:system1} is equivalent to~\eqref{for:system}.
	
	Suppose that $\hat{\bm{f}}$ is known model of the disturbances, then we can rewrite the equations and estimate $\hat{\bm{f}}-\bm{f}$ instead of $\bm{f}$. For instance, if we want to include the aerodynamic model of the disturbances after tilting a rotor, and estimate the difference with respect to this aerodynamic model, we can do it in this way. Here, we assume that we don't have such model and the problem is to estimate $\bm{f}$. 
\section{Learning-based control with Gaussian processes}\label{sec:lbc}
As introduced in the vehicle's dynamics~\eqref{for:system1}, we assume that parts of the dynamics  are unknown, i.e., $\bm{f}_v$ and $\f_\omega$. The proposed control strategy is based on design a controller by using a model that is updated by the predictions of a Gaussian process. 
In the following, we present the learning and control framework in detail.

\subsection{Learning with Gaussian processes}
For the compensation of the unknown dynamics of \eqref{for:system}, we use 
Gaussian Processes (GPs) to estimate the values of $\f_v,\f_\omega$ for a given state $\bm{q}$. For this purpose, $N(n):\N\to\N$ training points of the system \eqref{for:system} are collected to create a data set
\begin{align}
\D_{n(t)}=\{\bm{q}^{\{i\}},\y^{\{i\}}\}_{i=1}^{N(n)}.\label{for:dataset}
\end{align}
 The output data $\y\in\R^6$ is given by $\y=[(m\dot{\bm{v}}+\bm{\omega}\times m\bm{v}-mgRe_3-\bm{u}_2)^{\top} ,(J\dot{\bm{\omega}}+\bm{\omega}\times J\bm{\omega}-\bm{u}_1)^\top]^\top$ such that the first three components of $\y$ correspond to $\f_v$ and the remaining to $\f_\omega$.  The data set $\D_{n(t)}$ with $n\colon \R_{\geq 0}\to\N$ can change over time $t$, such that at time $t_1\in\R_{\geq 0}$ the data set $\D_{n(t_1)}$ with $N(n(t_1))$ training points exists. This allows to accumulate training data over time, i.e., the number of training points $N(n)$ in the data set $\D_n$ is monotonically increasing, but also ``forgetting" of training data to keep $N(n)$ constant. The time-dependent estimates of the GP are denoted by $\hat{\f}_{v,n}(\bm{q})$ and $\hat{\f}_{\omega,n}(\bm{q})$ to highlight the dependence on the corresponding data set $\D_{n}$. Note that this construction also allows offline learning, i.e. the estimation depends on previous collected data only, or any hybrid online/offline approach.

\begin{assum}\label{ass:2}
The number of data sets $\D_n$ is finite and there are only finitely many switches of~$n(t)$ over time, such that there exists a time $T\in\R_{\geq 0}$ where $n(t)=n_{\text{end}},\forall t\geq T$.
\end{assum}

\begin{rem}Note that Assumption \ref{ass:2} is little restrictive since the number of sets is often naturally bounded due to finite computational power or memory limitations and since the unknown functions $\f_v,\f_\omega$ in \eqref{for:system} are not time-dependent, long-life learning is typically not required. Therefore, there exists a constant data set $\D_{n_{end}}$ for all $t>T_{end}$. Furthermore, Assumption \ref{ass:2} ensures that the switching between the data sets is not infinitely fast which is natural in real world applications.\end{rem}

Gaussian process models have been proven as very powerful oracle for nonlinear function regression. For the prediction, we concatenate the~$N(n)$ training points of $\D_n$ in an input matrix~$X=[\bm{q}^1,\bm{q}^2,\ldots,\bm{q}^{N(n)}]$ and a matrix of outputs~$Y^\top=[\y^1,\y^2,\ldots,\y^{N(n)}]$, where $\y$ might be corrupted by additive Gaussian noise with $\mathcal{N}(0,\sigma I_6)$. Then, a prediction for the output $\y^*\in\R^6$ at a new test point $\bm{q}^*\in Q$ is given by
\begin{align}
	\mean_i(\y^*\vert \bm{q}^*,\D_n)&=m_i(\bm{q}^*)+\bm{k}(\bm{q}^*,X)^\top K^{-1}\label{for:gppred}\\
	&\phantom{=}\left(Y_{:,i}-[m_i(X_{:,1}),\ldots,m_i(X_{:,N})]^\top\right)\notag\\
	\var_i(\y^*\vert \bm{q}^*,\D_n)&=k(\bm{q}^*,\bm{q}^*)-\bm{k}(\bm{q}^*,X)^\top K^{-1} \bm{k}(\bm{q}^*,X).\notag
\end{align}
for all $i\in\{1,\ldots,6\}$, where $Y_{:,i}$ denotes the $i$-th column of the matrix of outputs~$Y$. The kernel $k\colon Q\times Q\to\R$ is a measure for the correlation of two states~$(\bm{q},\bm{q}^\prime)$, whereas the mean function $m_i\colon Q\to\R$ allows to include prior knowledge. The function~$K\colon Q^N\times  Q^N\to\R^{N\times N}$ is called the Gram matrix whose elements are given by $K_{j',j}= k(X_{:, j'},X_{:, j})+\delta(j,j')\sigma^2$ for all $j',j\in\{1,\ldots,N\}$ with the delta function $\delta(j,j')=1$ for $j=j'$ and zero, otherwise. The vector-valued function~$\bm{k}\colon Q\times  Q^N\to\R^N$, with the elements~$k_j = k(\bm{q}^*,X_{:, j})$ for all $j\in\{1,\ldots,N\}$, expresses the covariance between $\bm{q}^*$ and the input training data $X$. The selection of the kernel and the determination of the corresponding hyperparameters can be seen as degrees of freedom of the regression. The hyperparameters and the variance  $\sigma$ of the Gaussian noise in the training data can be estimated by optimizing the marginal log likelihood, see~\cite{rasmussen2006gaussian}. A powerful kernel for GP models of physical systems is the squared exponential kernel. An overview about the properties of different kernels can be found in~\cite{rasmussen2006gaussian}. In addition, the mean function can be achieved by common system identification techniques of the unknown dynamics~$\f_v,\f_\omega$ as described in~\cite{aastrom1971system}. However, without any prior knowledge the mean function is set to zero, i.e. $m_i(\bm{q})=0$.

Based on \eqref{for:gppred}, the normal distributed components~$y^*_i\vert \bm{q}^*,\D_n$ are combined into a multi-variable distribution $\y^*\vert(\bm{q}^*,\D_n) \sim \mathcal{N} (\Mean(\cdot),\Var(\cdot))$, where $ \Mean(\y^*\vert \bm{q}^*,\D_n)=[\mean_1(\cdot),\ldots,\mean_{6}(\cdot)]^\top$ and  $\Var(\y^*\vert \bm{q}^*,\D_n)=\diag\left[\var_1(\cdot),\ldots,\var_{6}(\cdot)\right]$. For simplicity, we consider identical kernels for each output dimension. However, the GP model can be easily adapted to different kernels for each output dimension.


For the later stability analysis of the closed-loop system, we introduce the following assumptions

\begin{assum}\label{ass:1}
	Consider a Gaussian process with the predictions $\hat{\f}_{v,n} \hbox{ and }\hat{\f}_{\omega,n}\in\mathcal{C}^0$ based on the data set $\D_n$~\cref{for:dataset}. Let $Q_\X\subset (SE(3)\times (\X\subset\R^6))$ be a compact set where $\hat{\f}_{v,n}, \hat{\f}_{\omega,n}$ are bounded on $\X$. There exists a bounded function $\bar{\rho}_n\colon Q_\X\to\R_{\geq 0}$ such that,  the prediction error is bounded by
	\begin{align}
		\Prob\Bigg\lbrace\left\Vert\begin{bmatrix} \f_v(\bm{s})- \hat{\f}_{v,n}(\bm{q}) \\ \f_\omega(\bm{q})- \hat{\f}_{\omega,n}(\bm{q})\end{bmatrix}\right\Vert\leq \bar{\rho}_n(\bm{q})\Bigg\rbrace\geq \delta
	\end{align}
	with probability $\delta\in (0,1]$, $\bm{q}\in Q_\X$.
\end{assum}
\begin{rem} Assumption \ref{ass:1} ensures that on each data set $\mathcal{D}_n$, there exists a probabilistic upper bound for the error between the prediction $\hat{\f}_{v,n}(\bm{q}),\hat{\f}_{\omega,n}(\bm{q})$ and the actual $\f_v(\bm{q}),\f_\omega(\bm{q})$ on a compact set.\end{rem}

To provide model error bounds, additional assumptions on the unknown parts of the dynamics \eqref{for:system} must be introduced \cite{wolpert1996lack}. 
\begin{assum}\label{as:rkhs}
	The kernel~$k$ is selected such that~$\f_v,\f_\omega$ have a bounded reproducing kernel Hilbert space (RKHS) norm on $Q_\X$, i.e.,~$\Verts{f_{v,i}}_k<\infty\hbox{ and }\Verts{f_{\omega,i}}_k<\infty$ for all~$i=1,2,3$.
\end{assum}
The norm of a function in a RKHS is a smoothness measure relative to a kernel~$k$ that is uniquely connected with this RKHS. In particular, it is a Lipschitz constant with respect to the metric of the used kernel. A more detailed discussion about RKHS norms is given in~\cite{wahba1990spline}.~\Cref{as:rkhs} requires that the kernel must be selected in such a way that the functions~$\f_v,\f_\omega$ are elements of the associated RKHS. This sounds paradoxical since this function is unknown. However, there exist some kernels, namely universal kernels, which can approximate any continuous function arbitrarily precisely on a compact set~\cite[Lemma 4.55]{steinwart2008support} such that the bounded RKHS norm is a mild assumption. Finally, with~\cref{as:rkhs}, the model error can be bounded as written in the following lemma.

\begin{lemma}[adapted from~\cite{srinivas2012information}]
	\label{lemma:boundederror}
Consider the unknown functions $\f_v,\f_\omega$ and a GP model satisfying~\cref{as:rkhs}. The model error is bounded  by
	\begin{align}
		\Prob\Bigg\lbrace\Bigg\Vert & \Mean\Bigg(\begin{bmatrix}\hat{\f}_{v,n}(\bm{q})\\ \hat{\f}_{\omega,n}(\bm{q})\end{bmatrix}\Bigg\vert \bm{q},\D_n\Bigg)-\begin{bmatrix}\f_v(\bm{q})\\\f_\omega(\bm{q})\end{bmatrix}\Bigg\Vert\notag\\
		&\hspace{1.5cm}\leq \Bigg\Vert\bm{\beta}_n^\top \Var^{\frac{1}{2}}\Bigg(\begin{bmatrix}\hat{\f}_n(\bm{q})\\ \hat{\f}_{\omega,n}(\bm{q})\end{bmatrix}\Bigg\vert \bm{q},\D_n\Bigg)\Bigg\Vert\Bigg\rbrace\geq \delta\notag
	\end{align}
	for~$\bm{q}\in Q_\X,\delta\in(0,1)$ with $\bm{\beta}_n\in\R^6$, \begin{align}
    (\bm{\beta}_n)_{j} =\sqrt{2\left\|\rho_{j}\right\|_{k}^{2}+300 \gamma_{j} \ln ^{3}\left(\frac{N(n)+1}{1-\delta^{1/6}}\right)}. 
\end{align}
The variable~$\gamma_{j} \in \R$ is the maximum information gain
\begin{align}
    \gamma_{j} &=\max _{\bm{q}^{\{1\}}, \ldots, \bm{q}^{\{N(n)+1\}} \in Q_\X} \frac{1}{2} \log \left|I+\sigma_{j}^{-2} K\left(\x, \x^{\prime}\right)\right| \\
    \x, \x^{\prime} & \in\left\{\bm{q}^{\{1\}}, \ldots, \bm{q}^{\{N(n)+1\}}\right\}.
\end{align}
\end{lemma}

\begin{proof}
It is a direct implication of~\cite[Theorem 6]{srinivas2012information}.
\end{proof}

Note that the prediction error bound in~\cref{ass:1} is given by 
$\bar{\rho}_n(\bm{q}):= ||\bm{\beta}_n^\top\Var^{\frac{1}{2}}([\hat{\f}_{v,n}(\bm{q})^\top,\hat{\f}_{\omega,n}(\bm{q})^\top]^\top\vert \bm{q},\D_n)||$ as shown by ~\cref{lemma:boundederror}.
\begin{rem}
An efficient algorithm can be used to find $\bm{\beta}_n$ based on the maximum information gain. Even though the values of the elements of~$\bm \beta$ are typically increasing with the number of training data, it is possible to learn the true function~$\f_v,\f_\omega$ arbitrarily exactly due to the shrinking variance $\Sigma$, see~\cite{Berkenkamp2016ROA}. In general, the prediction error bound $\bar{\rho}_n(\bm{q})$ is large if the uncertainty for the GP prediction is high and vice versa. Additionally, the bound is typically increasing if the set $Q_\X$ is expanded. The stochastic nature of the bound is due to the fact that just a finite number of noisy training points are available.
\end{rem}
\subsection{Control design}
Next we design the position and orientation controllers $\bm{u}_1$ and $\bm{u}_2$, respectively for the control system~\cref{for:system}. We prove the stability of the closed-loop with a proposed control law with multiple Lyapunov functions, where the $n$-th function is active when the GP predicts based on the corresponding training set $\D_n$. Note that due to a finite number of switching events, the switching between stable systems can not lead to an unbounded trajectory, ~\cite{liberzon1999basic}.

\textbf{Position controller:} Let $\bm{p}_d\in\mathbb{R}^3$ be the desired position. Define the position error by $\bm{e}=R^{T}(\bm{p}-\bm{p}_d)\in\mathbb{R}^{3}$. By differentiation the latter with respect to time, the error dynamics can be written as $\dot{\bm{e}}=-S(\bm{\omega})\bm{e}+\bm{v}$.

Let $\bm{z}\in\mathbb{R}^{3}$ be an error signal representing the difference between the desired and actual linear velocities, $\bm{z}=\bm{v}-\bm{v}_d\in\mathbb{R}^{3}$. Consider the Lyapunov function $V_{1,n}:\mathbb{R}^{3}\times\mathbb{R}^3\to\mathbb{R}_{\geq 0}$, $$V_{1,n}(\bm{e},\bm{z})=\frac{1}{2}||\bm{e}||^{2}+\frac{1}{2}m\bm{z}^{T}\bm{z}\geq 0$$ and note that $\Lambda_{1}||\bm{\zeta}||^2\leq V_{1,n}(\bm{e},\bm{z})\leq\Lambda_{2}||\bm{\zeta}||^2$, where $\bm{\zeta}=[\bm{e}^{T}, \bm{z}^{T}]^{T}$, $\Lambda_{1}=\frac{1}{2}\hbox{max}\{1,m\}$, $\Lambda_{2}=\frac{1}{2}\hbox{min}\{1,m\}$. 

By differentiating $V_{1,n}$ with respect to the time along the trajectories of the system, using the expresions for $\bm{e}, \dot{\bm{e}}$, $\bm{y}, \dot{\bm{y}}$, the fact that $S(\bm{\omega})$ is skew-symmetric, and \eqref{for:system1}, we obtain that \begin{align}\label{vdot}\dot{V}_{1,n}(\bm{e},\bm{z})=&\bm{e}^{\top}\bm{v}_d\\&+\bm{z}^{\top}\{\bm{e}-S(\bm{\omega})m\bm{v}+{\bm{f}}_{v}+\bm{u}_2-m\dot{\bm{v}}_d\}.\nonumber\end{align} We design the desired velocity as $\bm{v}_d=-k_1\bm{e}$, and the position controller as \begin{align}\label{positioncontrol}
    \bm{u}_2=&-k_2\bm{z}-\bm{e}+S(\bm{\omega})m\bm{v}\\&-k_1m(S(\bm{\omega})\bm{e}+\bm{v})-\mu({\bm{f}}_{v,n}\mid \bm{q},\mathcal{D}_n),\nonumber
\end{align}where $k_1,k_2\in\mathbb{R}_{>0}$ are controller gains to be tuned.

\begin{prop}
Consider the system ~\eqref{for:system1} and a GP model trained with ~\eqref{for:dataset} satisfying ~\cref{ass:2}, \ref{ass:1}, \ref{as:rkhs}. The position control law \eqref{positioncontrol} guarantees that the tracking error $\bm{\zeta}$ is uniformly ultimately bounded in probability by 
\begin{align}
\Prob\left\lbrace\Verts{\bm{\zeta}(t)}\leq \sqrt{\frac{\Lambda_{2}}{\Lambda_{1}}}\max_{\bm{q}\in Q_\X}\bar{\rho}_{n_\text{end}}(\bm{q}),\forall t\geq T\right\rbrace\geq\delta
\end{align}
on $Q_\X$ with $T\in\R_{\geq 0}$, and exponentially converges to zero.
\end{prop}

\begin{proof}
By substituting $\bm{v}_d=-k_1\bm{e}$ and $\bm{u}_2$ given by  \eqref{positioncontrol} into \eqref{vdot} we get $$\dot{V}_{1,n}(\bm{e},\bm{z})\leq -\min\{k_1,k_2\}||\bm{\zeta}||^2+\bm{z}^{T}({\bm{f}}_{v,n}-\mu( {\bm{f}}_{v} \mid \bm{q},\mathcal{D}_n)).$$ So, by ~\cref{lemma:boundederror}, the evolution of the Lyapunov function $V_{1,n}$ can be upper bounded by 
\begin{align}
\Prob\{\dot{V}_{1,n}&\leq -\min\{k_1,k_2\}||\bm{\zeta}||^2
+||\bm{z}||\bar{\rho}_n(\bm{q})\}\geq\delta\notag.\label{for:UPV3}
\end{align}
 Thus, the evolution is negative with probability $\delta$ for all $\bm{\zeta}$ such that $\displaystyle{\Verts{\bm{\zeta}}>\max_{\bm{q}\in Q_\X}\bar{\rho}_n(\bm{q})\frac{1}{\min\{k_ 1,k_2\}}}$, where a maximum of $\bar{\rho}_n$ exists regarding to~\cref{ass:1}. Finally, the Lyapunov function $V_{1,n}$ is lower and upper bounded by $\alpha_1(\Verts{\bm{\zeta}})\leq V_{1,n}(\bm{\zeta})\leq\alpha_2(\Verts{\bm{\zeta}})$, where $\alpha_1(r)=\Lambda_{1}r^2$ and $\alpha_2(r)=\Lambda_{2}r^2$. Thus, we can compute the maximum tracking error $b_n\in\R_{\geq_0}$ such that $\Prob\{||\bm{\zeta}||\leq b_n\}\geq\delta$ by 
$
b_n=\sqrt{\frac{\Lambda_{2}}{\Lambda_{1}}}\max_{\bm{q}\in Q_\X}\bar{\rho}_n(\bm{q})$.

Since~\cref{ass:2} only allows a finite number of switches, there exists a time $T_{\text{end}}\in\R_{\geq 0}$ such that $n(t)=n_{\text{end}}$ $\forall t\geq T_{\text{end}}$ and so $\mathcal{D}_{n(T_{\text{end}})}=\mathcal{D}_{n_{\text{end}}}$. Thus, there exists $T\geq T_{\text{end}}$ such that  $\Prob\{\Verts{\bm{\zeta}(t)}\leq\sqrt{\frac{\Lambda_{2}}{\Lambda_{1}}} \max_{\bm{q}\in Q_\X}\bar{\rho}_{n_\text{end}}(\bm{q}),\forall t\geq T\}\geq\delta$.\end{proof}

\textbf{Attitude controller:} Let $R_d\in SO(3)$ be the desired attitude of the vehicle. Define the real-valued error $\Psi:SO(3)\times SO(3)\to [0,2]$ by $\Psi(R,R_d)=\frac{1}{2}\hbox{Tr}[I_{3\times 3}-R_d^{T}R]$. This function is locally positive definite about $R=R_d$ within the region where the rotation angles between $R$ and $R_d$ is less than $\pi$ rads. This can be represented by the set \begin{equation}\label{invariantset}
\mathcal{L}=\{R, R_{d}\in SO(3):\Psi(R,R_d)<2\}
\end{equation} which almost covers $SO(3)$. Note that for any rotation matrix $Q=R_d^\top R\in SO(3)$ it holds that $-1\leq\hbox{Tr}(Q)\leq 3$ and $\hbox{Tr}(Q)=3$ if and only if $Q=I_{3\times 3}$. When $\hbox{Tr}(Q)=-1$ or when equivalently $R=R_d\hbox{exp}[\pm\pi\hat{q}]$ for every $\hat{q}\in S^2$ we have the case that $\Psi(R,R_d)<2$. Thus $\Psi(R,R_d)\in[0,2]$ and $\Psi(R,R_d)=0$ only when $R=R_d$.

Let us define the vector $\bm{\chi}:SO(3)\times SO(3)\to\mathbb{R}^3$ by $$\bm{\chi}(R,R_d)=\frac{1}{2}S^{-1}(R_d^\top R-R^\top R_d).$$ By using Rodriguez formula it can be shown (see \cite{lee2012exponential}) that there are positive constants $0<c_1<c_2$ such that $c_1||\bm{\chi}||^2\leq\Psi(R,R_d)\leq c_2||\bm{\chi}||^2$. 

Next, we compute the critical points of $\Psi$ in $\mathcal{L}$. Note that \begin{align*}
\dot{\Psi}(R,R_d)&=-\frac{1}{2}\hbox{Tr}(R_d^\top\dot{R})=-\frac{1}{2}\hbox{Tr}(R_d^\top RS(\bm{\omega}))=\\&=\frac{1}{2}\bm{\omega}^\top S^{-1}(R_d^\top R-R^\top R_d)=\bm{\omega}^\top\bm{\chi}.
\end{align*} So, critical points of $\Psi$ on $\mathcal{L}$ are given by $R_d^\top R-R^\top R_d=0$, that is, $R=R_d$.

Denote by $\bm{\Omega}$ the signal error representing the difference between the desired and actual angular velocities, that is, $\bm{\Omega}=\bm{\omega}-\bm{\omega}_d\in\mathbb{R}^3$ and consider the Lyapunov function \begin{equation}\label{Lfunction2}
V_{2,n}(\bm{\chi},\bm{\Omega})=\Psi(R,R_d)+\frac{1}{2}\bm{\Omega}^\top J\bm{\Omega}.
\end{equation} Note that $V_{2,n}$ is positive definite and there exists positive constants $K_1,K_2\in\R_{\geq_0}$ such that $K_1||\bm{\eta}||^2\leq V_{2,n}(\bm{\chi},\bm{\Omega})\leq K_{2}||\bm{\eta}||^2$, where $\bm{\eta}=[\bm{\chi}^{T}, \bm{\Omega}^{T}]^{T}$, $K_1=\hbox{max}\{c_2,\frac{1}{2}\lambda_{\hbox{max}}(J)\}$, $K_{2}=\hbox{min}\{c_1,\frac{1}{2}\lambda_{\hbox{min}}(J)\}$. 

By differentiating $V_{2,n}$ with respect to the time along the trajectories of the system we have

\begin{align}\label{2vdot}\dot{V}_{2,n}(\bm{\chi},\bm{\Omega})=&\dot{\Psi}(R,R_d)+\bm{\Omega}^\top J\dot{\bm{\Omega}}=\bm{\omega}^{\top}\bm{\chi}+\bm{\Omega}^\top J(\dot{\bm{\omega}}-\dot{\bm{\omega}}_d)\nonumber
\\=&\bm{\omega}^{\top}_d\bm{\chi}+\bm{\Omega}^\top\bm{\chi}+\bm{\Omega}^\top(J\dot{\bm{\omega}}-J\dot{\bm{\omega}}_d)\nonumber\\
=&\bm{\omega}^{\top}_d\bm{\chi}+\bm{\Omega}^\top (\bm{\chi}-J\bm{\dot{\omega}}_d-\bm{\omega}\times J\bm{\omega}\\ & \hspace{2cm}+{\bm{f}}_{\omega}+\bm{u}_1).\nonumber\end{align} We design the desired velocity as $\bm{\omega}_d=-k_3\bm{\chi}$, and the attitude controller as \begin{align}\label{attitudecontrol}
    \bm{u}_1=&-k_4\bm{\Omega}-k_5\bm{\chi}+\bm{\omega}\times J\bm{\omega}\\&-k_3J\dot{\bm{\chi}}-\mu(\bm{f}_{w,n}\mid \bm{q},\mathcal{D}_n),\nonumber
\end{align}where $k_3,k_4,k_5\in\mathbb{R}_{>0}$ are controller gains to be tuned. 

\begin{prop}\label{prop2}
Consider the system ~\eqref{for:system} and a GP model trained with ~\eqref{for:dataset} satisfying ~\cref{ass:2}, \ref{ass:1}, \ref{as:rkhs}. The orientation control law \eqref{attitudecontrol} guarantees that the tracking error $\bm{\eta}$ is uniformly ultimately bounded in probability by 
\begin{align}
\Prob\{\Verts{\bm{\eta}(t)}\leq\sqrt{\frac{K_{2}}{K_{1}}} \max_{\bm{q}\in Q_\X}\bar{\rho}_{n_\text{end}}(\bm{q}),\forall t\geq T\}\geq\delta
\end{align}
with time constant $T\in\R_{\geq 0}$ on $Q_\X$.
\end{prop}

\begin{proof}
Similarly as with the position controller, by substituting $\bm{\omega}_d=-k_3\bm{\chi}$, $\bm{u}_1$ given by  \eqref{attitudecontrol} into \eqref{2vdot} and by employing Pauli inequality and ~\cref{lemma:boundederror}, the evolution of the Lyapunov function $V_{2,n}$ can be upper bounded by \begin{align}
\Prob\{\dot{V}_{2,n}&\leq -\min\{k_3,k_4, k_6\}||\bm{\eta}||^2
+||\bm{\Omega}||\bar{\rho}_n(\bm{q})\}\geq\delta\notag,\label{for:UPV3}
\end{align} where $k_6=\hbox{min}\{\frac{(k_5-1)\rho}{2},\frac{(k_5-1)}{2\rho}\}$ for some $\rho>0$ such that $\frac{(k_5-1)\rho}{2},\frac{(k_5-1)}{2\rho}$ are positive constants.

 Thus, the evolution is negative with probability $\delta$ for all $\bm{\eta}$ such that \begin{align}
\Verts{\bm{\eta}}>\max_{\bm{q}\in Q_\X}\bar{\rho}_n(\bm{q})\frac{1}{\min\{k_ 3,k_4,k_6\}},\notag
\end{align}
where a maximum of $\bar{\rho}_n$ exists regarding to~\cref{ass:1}. Finally, the Lyapunov function $V_{2,n}$ is lower and upper bounded by $\overline{\alpha}_1(\Verts{\bm{\eta}})\leq V_{2,n}(\bm{\eta})\leq\overline{\alpha}_2(\Verts{\bm{\eta}})$, where $\overline{\alpha}_1(r)=K_{1}r^2$ and $\alpha_2(r)=K_{2}r^2$. Thus, we can compute the maximum tracking error $c_n\in\R_{\geq_0}$ such that $\Prob\{||\bm{\eta}||\leq c_n\}\geq\delta$ by 
$
c_n=\sqrt{\frac{K_{2}}{K_{1}}}\max_{\bm{q}\in Q_\X}\bar{\rho}_n(\bm{q})$.

Since~\cref{ass:2} only allows a finite number of switches, there exists a time $T_{\text{end}}\in\R_{\geq 0}$ such that $n(t)=n_{\text{end}}$ $\forall t\geq T_{\text{end}}$ and so $\mathcal{D}_{n(T_{\text{end}})}=\mathcal{D}_{n_{\text{end}}}$. Thus, there exists $T\geq T_{\text{end}}$ such that  $\Prob\{\Verts{\bm{\eta}(t)}\leq\sqrt{\frac{K_{2}}{K_{1}}} \max_{\bm{q}\in Q_\X}\bar{\rho}_{n_\text{end}}(\bm{q}),\forall t\geq T\}\geq\delta$.\end{proof}

\begin{cor} Under the conditions of Proposition $2$, if in addition we assume that $\Psi(R(0),R_d)<2$, $\displaystyle{||\bm{\Omega}(0)||^{2}<\frac{1}{\lambda_{\hbox{min}}(J)}(2-\Psi(R(0),R_d))}$, then $\mathcal{L}$ is positively invariant.
\end{cor}

\begin{proof}
Note that $$\dot{V}_{2,n}(\bm{\chi},\bm{\Omega})\leq -\min\{k_3,k_4\}||\bm{\eta}||^2-\bm{\Omega}^\top(\bm{\f}_{\omega}-\mu(\bm{f}_{\omega,n}\mid \bm{q},\mathcal{D}_n)).$$  Moreover, 
\begin{align*}
    \dot{V}_{2,n}(\bm{\chi},\bm{\Omega})&\leq -\frac{\min\{k_3,k_4\}}{c_2}V_{2,n}(\bm{\chi},\bm{\Omega})\\&+\frac{1}{c_2}\bm{\Omega}^\top(\bm{f}_{\omega}-\mu(\bm{f}_{\omega,n}\mid \bm{q},\mathcal{D}_n))V_{2,n}(\bm{\chi},\bm{\Omega})\\
    &=\frac{1}{c_2} (\bm{\Omega}^\top(\bm{f}_{\omega}-\mu(\bm{f}_{\omega,n}\mid \bm{q},\mathcal{D}_n))\\
    &-\min\{k_3,k_4\})V_{2,n}
\end{align*}
which implies that we can choose $k_3,k_4,c_2$ such that for every $t\geq 0$ it holds $$V_{2,n}(\bm{\chi}(t),\bm{\Omega}(t))\leq V_{2,n}(\bm{\chi}(0),\bm{\Omega}(0)),$$ or equivalently
\begin{align}
    \Psi(R(t),R_d)\leq&\Psi(R(0),R_d)+\frac{1}{2}\bm{\Omega}(0)^\top J\bm{\Omega}(0)\\
    \leq &\Psi(R(0),R_d)+\frac{1}{2}\lambda_{\hbox{min}}(J)||\bm{\Omega}(0)||^2\\
    \leq&\Psi(R(0),R_d)+2-\Psi(R(0),R_d)<2.
\end{align}
Thus, $\mathcal{L}$ is positively invariant and then $R(t)\in\mathcal{L}$, $\forall t$.
\end{proof}
\begin{rem}
Note the orientation controller $ \bm{u}_1$ provides \textit{almost} global exponential convergence for orientation stabilization since the initial conditions should satisfy $\Psi(R(0),R_d)<2$, $\displaystyle{||\bm{\Omega}(0)||^{2}<\frac{1}{\lambda_{\hbox{min}}(J)}(2-\Psi(R(0),R_d))}$.
\end{rem}

\section{Experimental results}\label{sec:exp}

In this section the control algorithm proposed above is applied to a fault tolerant hexarotor vehicle. The main objective is to show that the GP  estimates allow to improve the performance of the control after a failure. 

The structure of the UAV in Fig. \ref{fig:vuelo_outdoor} is based on a DJI F550 Flame Wheel ARF Kit. The rotors are DJI 2212/920KV Brushless DC motors capable of providing \SI{1}{\kilogram} thrust each, while the total weight of the vehicle sits around \SI{2.8}{\kilogram}. This allows to maintain the vehicle in the air even when one of the rotors is not providing any thrust. 
The UAV is controlled by means of a custom made flight computer, with a Cortex M3 microcontroller. It includes a variety of sensors, such as an MPU600 Inertial Measurement Unit and a HMC5883L digital compass to estimate the attitude of the vehicle. 

The hexarotor design was proposed in \cite{TMECH2020}. It is equipped with a mechanism that allows to instantaneously tilt one rotor when a failure occurs, increasing in this was the sate of achievable torque after a rotor failure, when comparing with other fault-tolerant hexarotor designs \cite{giribet2016analysis,Michieletto2017}. Figure \ref{fig:control_alloc} shows a diagram block of how the fault compensation system works. The control algorithm computes the control signal $u_{cmd}$ (torque and forces), and the control allocation computes the PWM signal commanded to each rotor to achieve the commanded torque and forces. When a rotor fails, the fault detection system is activated and a signal $M_i$ is commanded to the hexarotor indicating which rotor failed, then a device is activated tilting one of the remaining rotors. Signal $M_i$ is also received by the control allocation system, which takes into account the reconfiguration of the remaining functioning rotors. As it can be seen, the control algorithm doesn't need to be changed. Next, the fault detection and control allocation subsystems are considered as part of the control algorithm. 

\begin{figure}[t]
    \centering
    \includegraphics[width = 0.9\columnwidth]{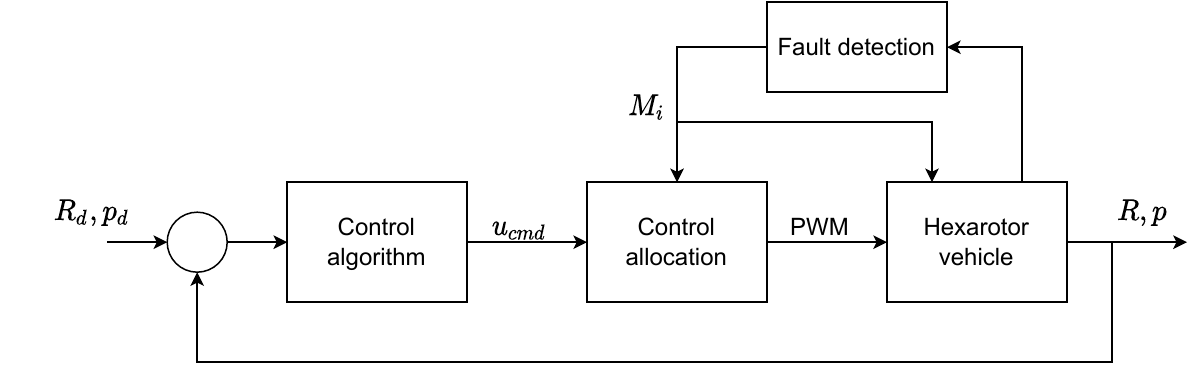}
    \caption{Architecture of the fault detection and control allocation subsystems.}
    \label{fig:control_alloc}
\end{figure}

Figure \ref{fig:vuelo_outdoor} shows the behavior of the hexarotor during a failure. It can be appreciated how the vehicle recovers after detecting the failure. Also, in frame (d) it can be appreciated how one rotor was reconfigured after the failure. 

Although the control algorithm is the same, experimentally can be noted that the performance of the control system is degraded respect to the nominal state. As it was mentioned before, there are several reasons why the performance is not the same. The reconfiguration of rotors has impact on the vehicle dynamic \cite{UnmSyst2020}. In Fig. \ref{fig:path-tracking} it can be appreciated how the the performance of the hexarotor is affected after the failure. This is an indoor experiment, in red is given the commanded trajectory and in blue the true position measured with a Marvelmind, at $50ms$ with precision $\pm 2 cm$. It can be noted when a fail in motor 3 ($M_3$) is activated, and then how the tracking performance is degraded. A video for this experiment can be seen at \cite{VideoPathFollow}. 

\begin{figure}[t]
    \centering
    \includegraphics[width = 0.9\columnwidth]{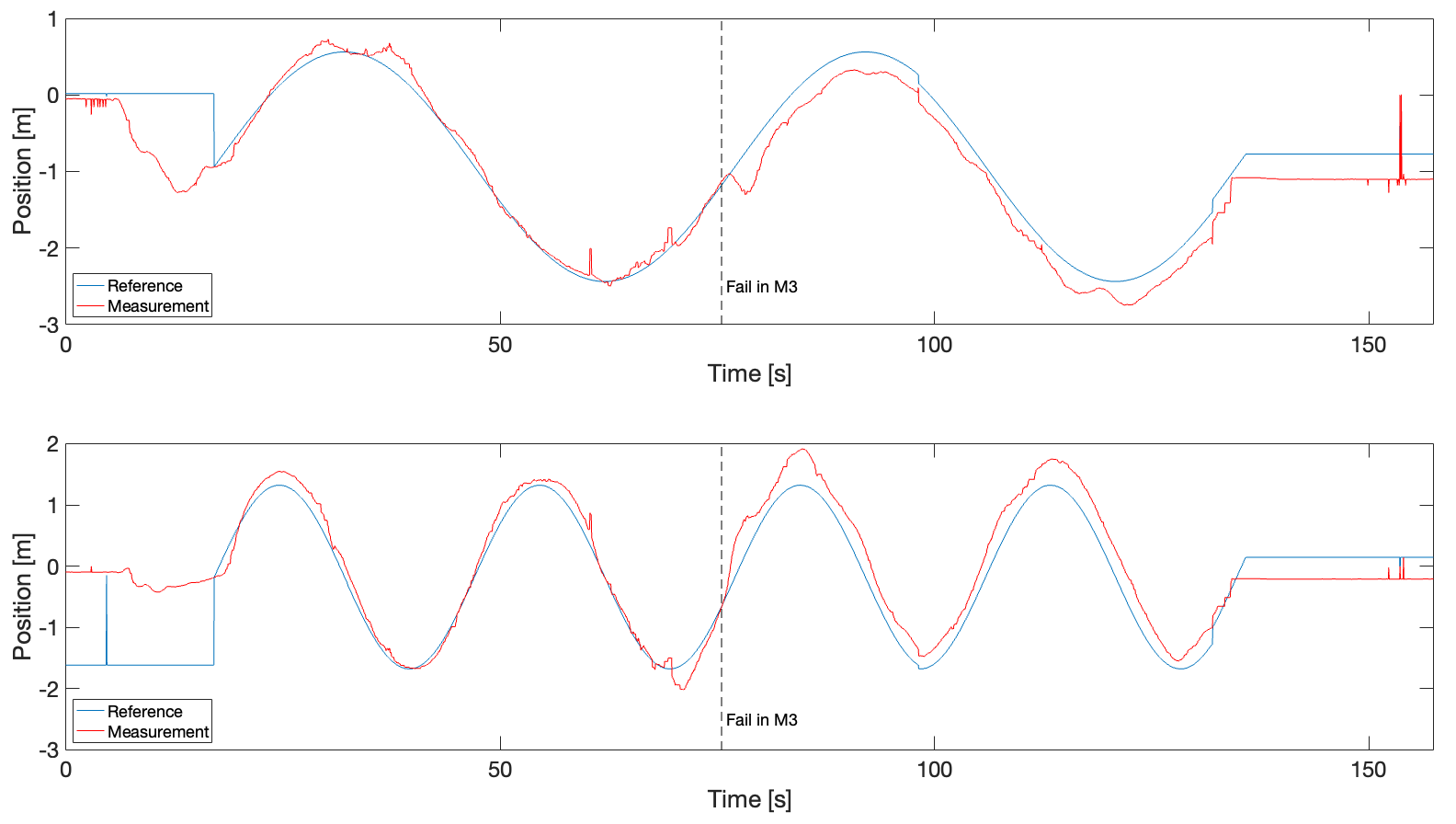}
    \caption{Trajectory tracking with a fault-tolerant hexarotor.}
    \label{fig:path-tracking}
\end{figure}

In Fig. \ref{fig:GPest} it can be shown that, after the failure, the value $\hat{\f}_{v,n}(\bm{q})$ detects a change in the disturbances, which is consistent with a degradation in the control system performance. The source of this perturbation is not clear, model uncertainty, for instance an error in the tilting angle, aerodynamics perturbations, among others could be affecting the vehicle. But, it is not relevant here to determine the source or sources, but estimate the resultant effect.
In this example, the perturbation estimations were not used as a feedback in the control loop.

\begin{figure}[t]
    \centering
    \includegraphics[width = 0.9\columnwidth]{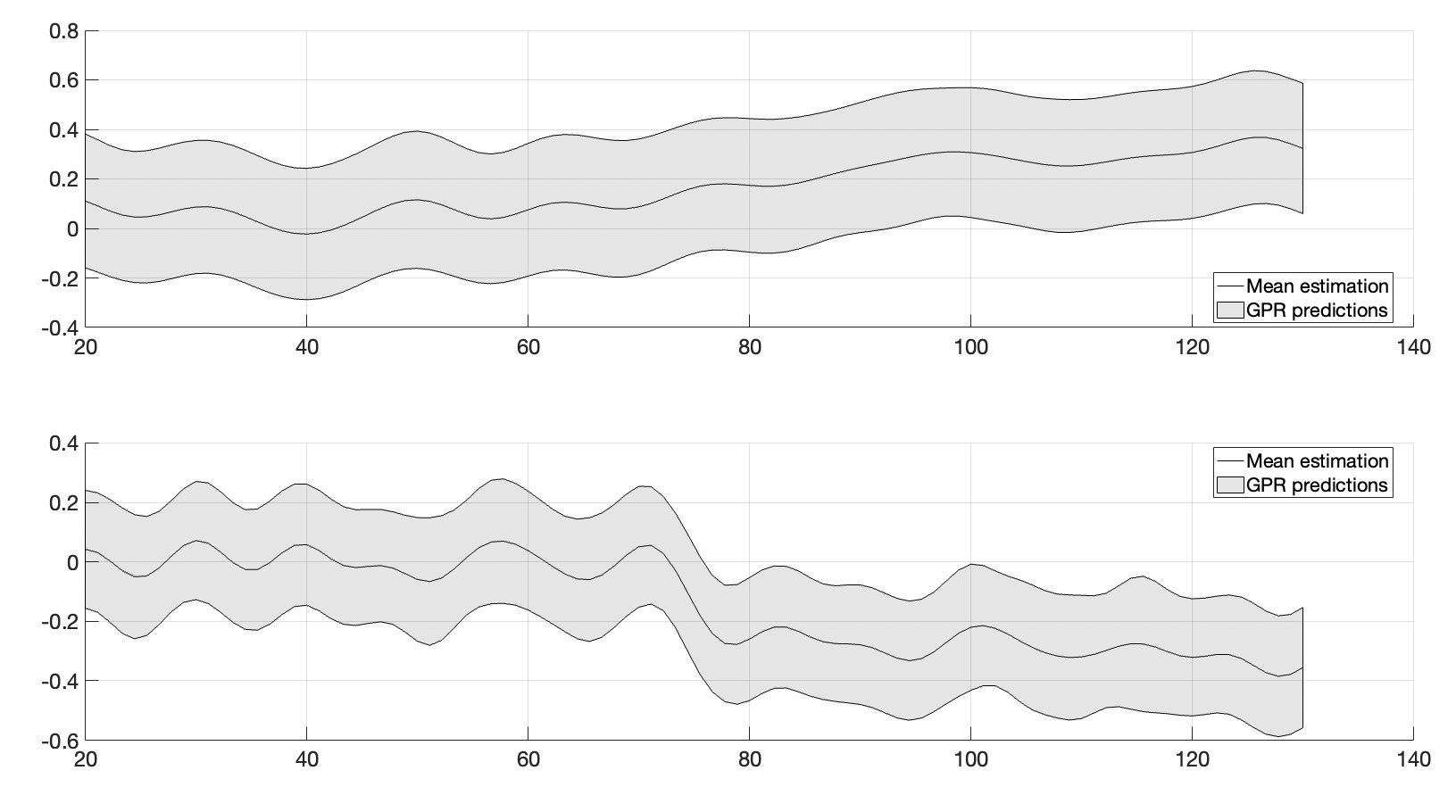}
    \caption{GP estimation of the horizontal components of the perturbations $\hat{\f}_{v,n}$. The solid line is the mean and the shadow represents the 95\% prediction interval.}
    \label{fig:GPest}
\end{figure}


To test how the control algorithm performs when the GP estimates are used to compensate the disturbances a second experiment was carried out. 
Figure \ref{fig:controllers_diagram} shows a block diagram of the control algorithm architecture proposed here. Two nested controllers are used for position and attitude control. The attitude control is executed at a frequency of \SI{200}{\hertz} and the position control at \SI{20}{\hertz}. 

\begin{figure}[t]
    \centering
    \includegraphics[width = \columnwidth]{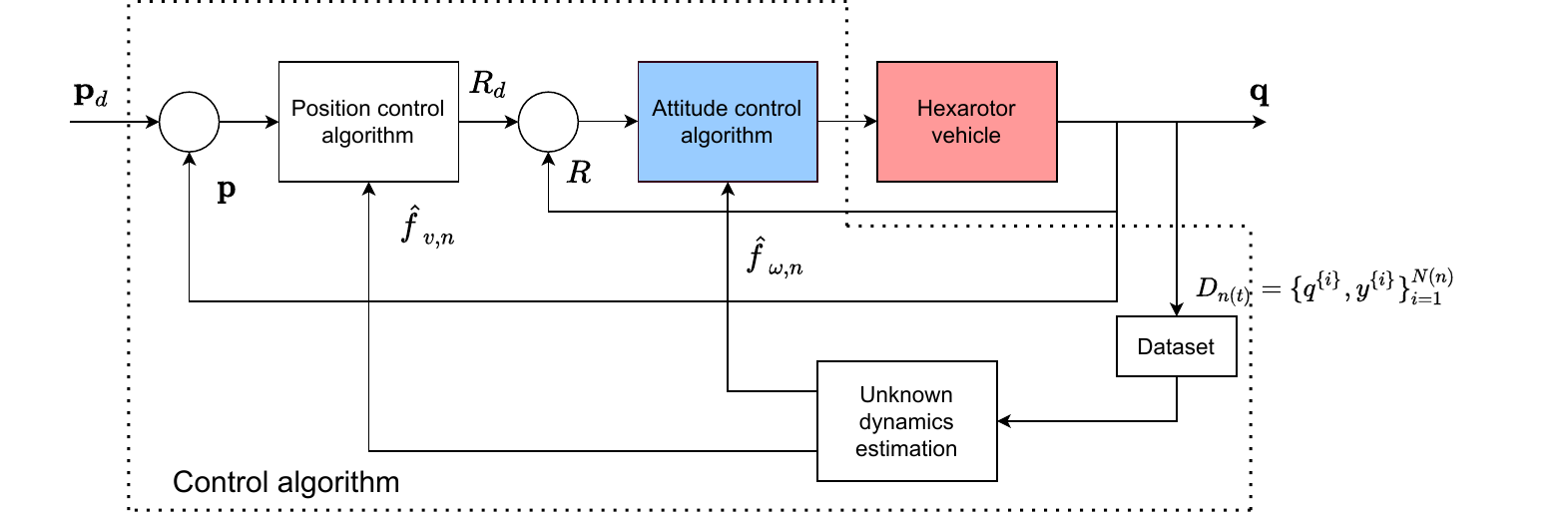}
    \caption{Control scheme diagram. The fault detection and control allocation subsystems are not included here.}
    \label{fig:controllers_diagram}
\end{figure}

Two different flights were performed in similar conditions, with a trajectory as in Fig. \ref{fig:path-tracking}, and a rotor failure is injected during the flight. A video of these tests can be found at \cite{VideoGPbased}.
In the first flight  the dataset was stored but the control algorithm was not compensated with the GP estimates. Figure \ref{fig:first-flight} shows the attitude response (pitch and roll) of the vehicle during the flight. It can be noted that, at approximately $35sec$ a failure is introduced and the vehicle recovers stability after a rotor reconfiguration., this can be particularly noticed in Fig. \ref{fig:first-flight-pwm}, where the PWM signals commanded to each rotor are shown. The second flight (see Fig. \ref{fig:second-flight}) incorporated the corrections provided by the GP estimates in the control loop. It can be noted that after the failure (approximately at $33sec$) the performance of the vehicle improves with respect to the first flight, $17\%$ in terms of the mean square error respect to the first flight. Also it can be noticed comparing the PWM signals given in Fig. \ref{fig:first-flight-pwm} and \ref{fig:second-flight-pwm}, that the signals without the compensation are more noisy.

\begin{figure}[t]
    \centering
    \includegraphics[width = \columnwidth]{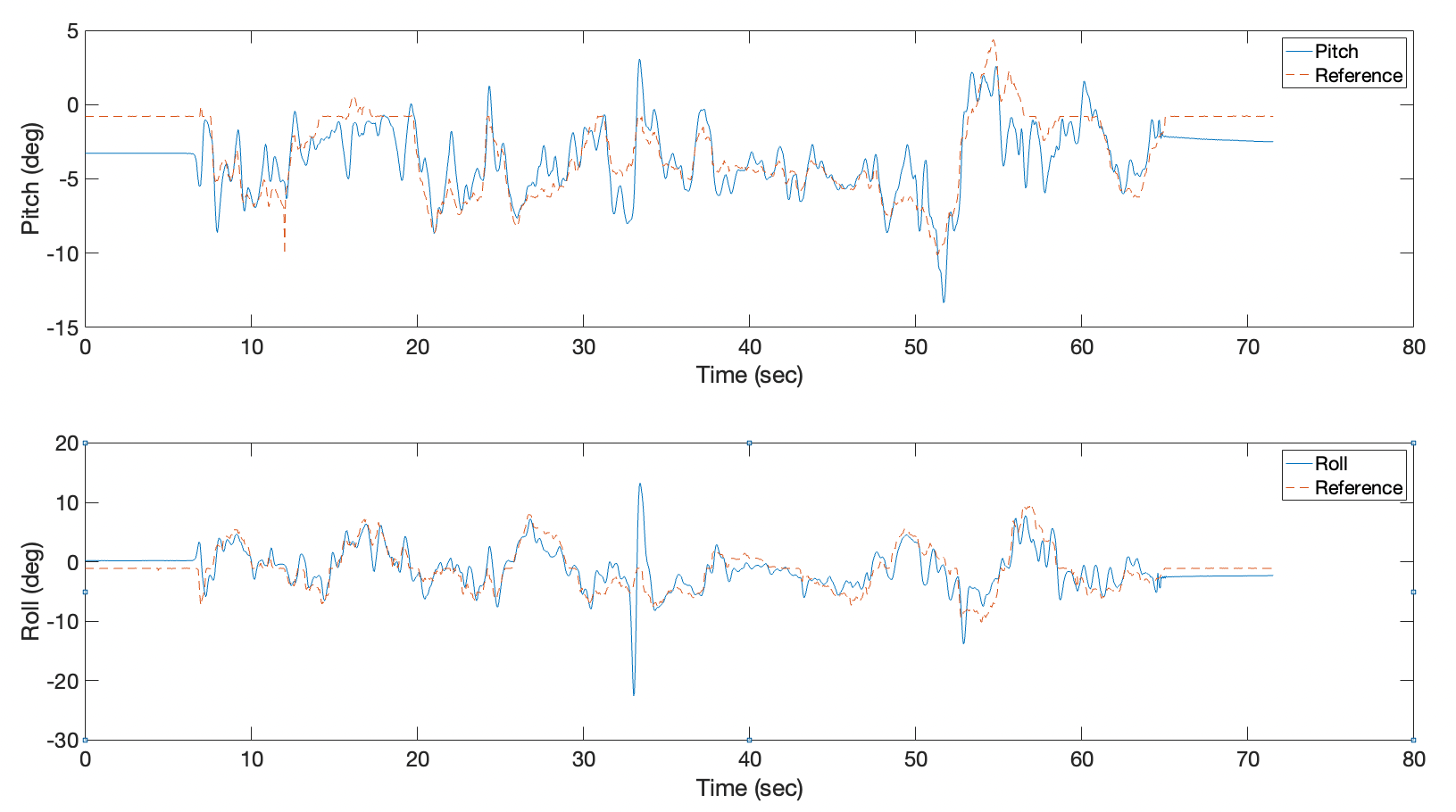}
    \caption{Orientation of the hexarotor, for a failure occurring during the flight.}
    \label{fig:first-flight}
\end{figure}

\begin{figure}[t]
    \centering
    \includegraphics[width = 0.9\columnwidth]{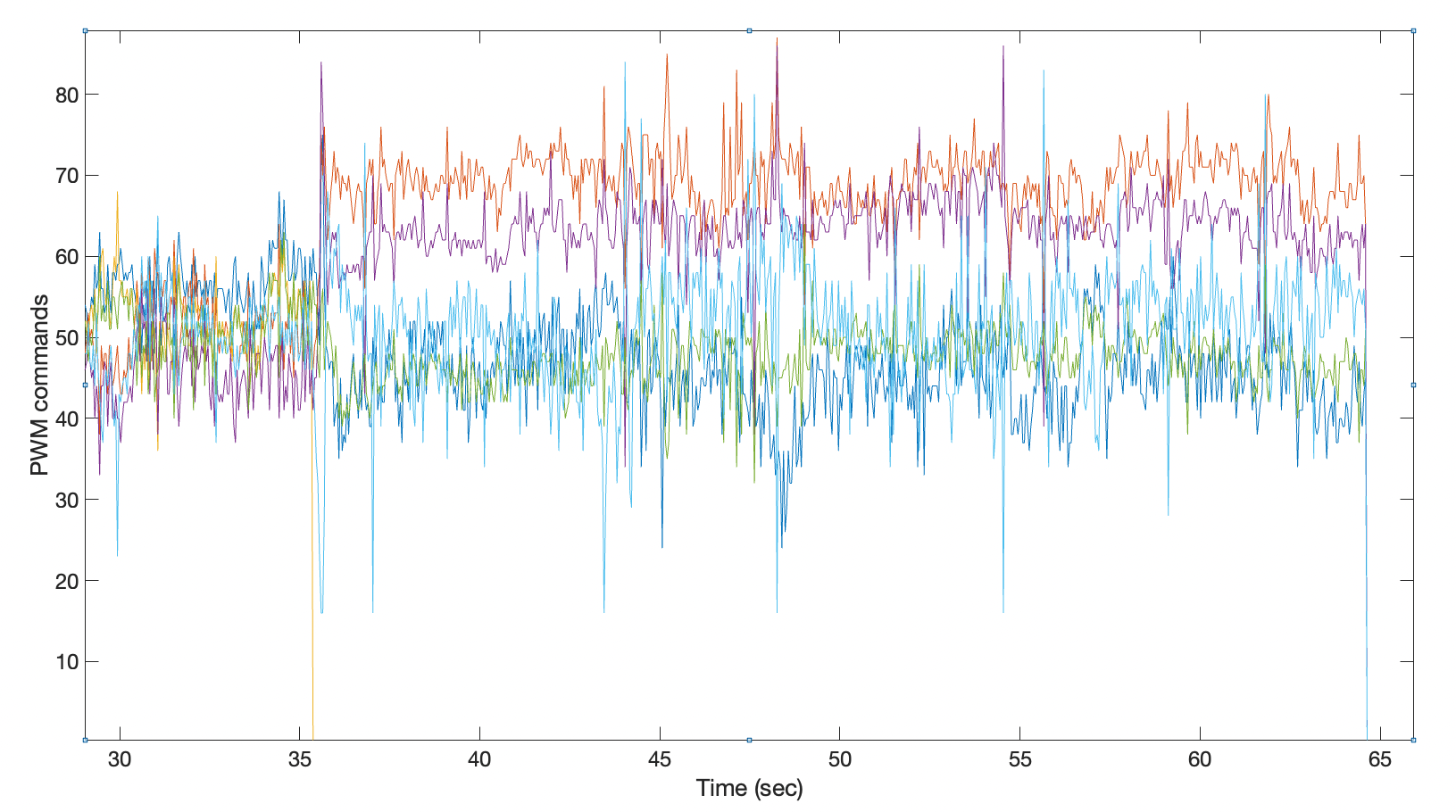}
    \caption{PWM signals of the vehicle, for a failure occurring during the flight.}
    \label{fig:first-flight-pwm}
\end{figure}

\begin{figure}[t]
    \centering
    \includegraphics[width = \columnwidth]{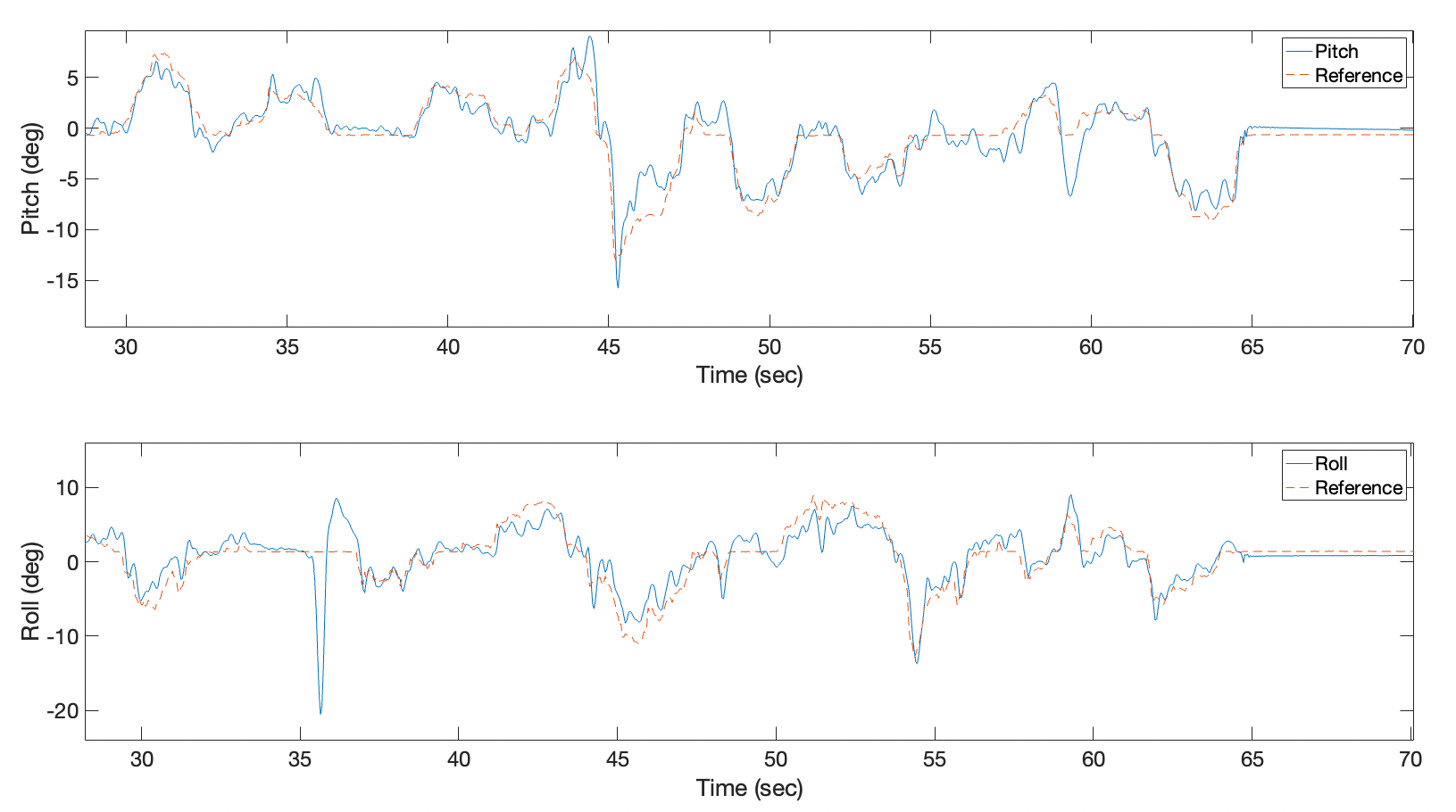}
    \caption{Orientation of the hexarotor, for a failure occurring during the flight. Vehicle with GPS estimates compensation.}
    \label{fig:second-flight}
\end{figure}

\begin{figure}[t]
    \centering
    \includegraphics[width = 0.9\columnwidth]{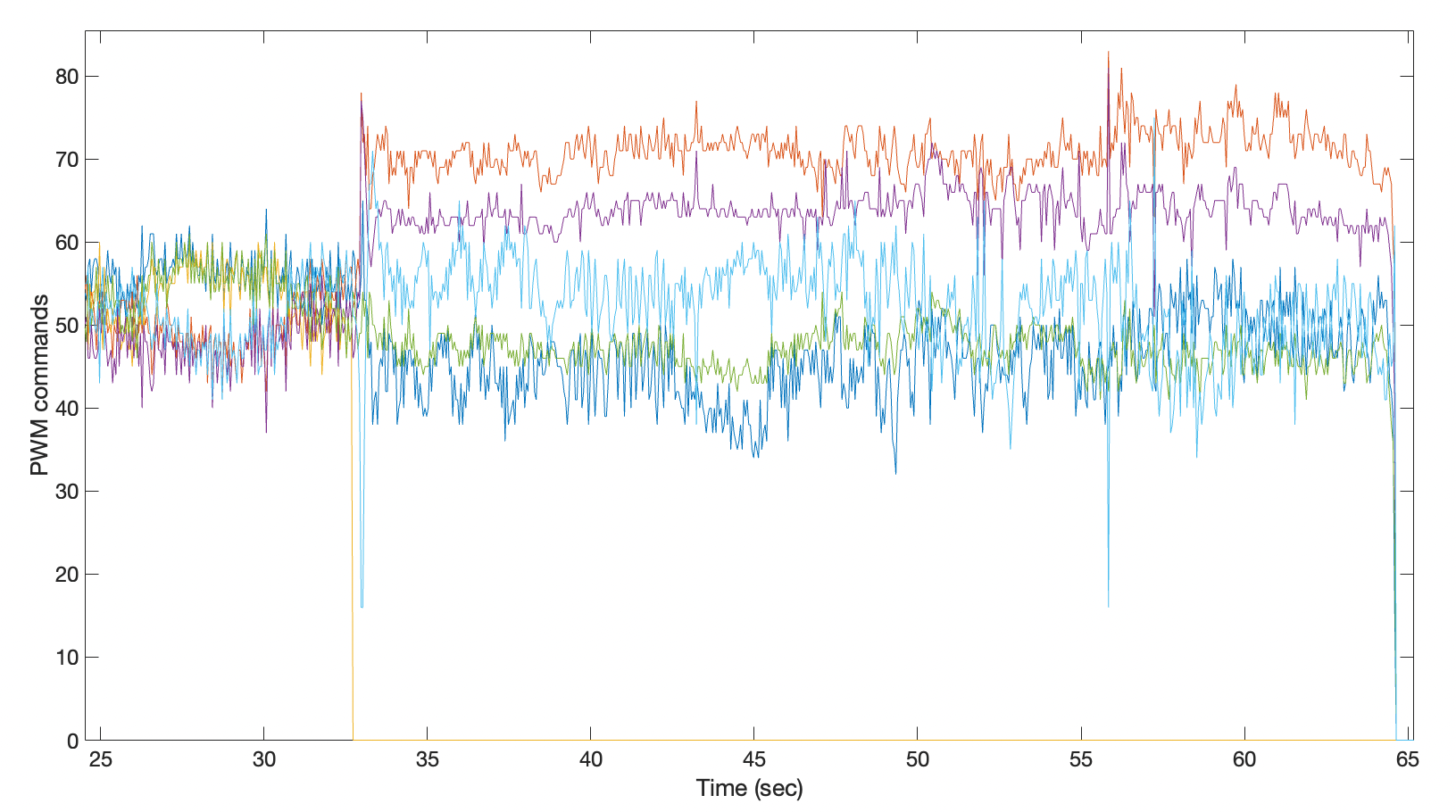}
    \caption{PWM signals of the vehicle, for a failure occurring during the flight. Vehicle with GPS estimates compensation.}
    \label{fig:second-flight-pwm}
\end{figure}

\section*{Conclusion}
We present a learning-based fault-tolerant control law for an hexarotor UAV under model uncertainties using GP to predict the unknown dynamics. In particular, the use of a learning strategy based on GP is studied to estimate certain uncertainties that appear in a hexacopter vehicle with the ability to reconfigure its rotors to compensate for failures. The rotors reconfiguration introduces disturbances that make the dynamic model of the vehicle differ from the nominal model. Several factors can introduce these disturbances, e.g. errors in the angles at which the actuators are tilted, characterizations of the motors in different work regimes, and even aerodynamic disturbances caused by a repositioning of the rotors. From the point of view of compensating for these disturbances, it is not necessary to discern where the disturbances come from. The proposed control law learn the uncertainties in the model after a failure is detected guaranteeing the probabilistic boundedness of the tracking error to the reconfigured attitude and positions with high probability. The system obtains information in a data set which uses to estimate the uncertainties with a GP model, this allows to improve the model and, thus, mitigate the model uncertainties. Based on experimental flights, it was observed that the algorithm allows an improvement of the performance of the system, for which two flights were carried out under similar conditions, compensating in one case the control loop with the estimates provided by the GP.

\section*{Acknowledgment}
L. Colombo is very grateful to T. Beckers from University of Pennsylvania for many useful comments and stimulating discussions on learning control with Gaussian processes.
\bibliography{mybib}
\bibliographystyle{ieeetr}

\end{document}

%% file: tikz.tex
\usepackage{tikz}
\usepackage{pgfmath}
\usetikzlibrary{babel}
\usetikzlibrary{intersections}
\usetikzlibrary{fadings}
\usetikzlibrary{calc}
\usetikzlibrary{math}
\usetikzlibrary{fit}
\usetikzlibrary{positioning}
\usetikzlibrary{patterns}
\usetikzlibrary{shapes, arrows, arrows.meta}
\usetikzlibrary{shapes.geometric, shapes.multipart}
\usetikzlibrary{decorations.pathreplacing}
\usetikzlibrary{decorations.pathmorphing}
\usetikzlibrary{decorations.text}
\usetikzlibrary{backgrounds}
\pgfdeclarelayer{background}
\pgfdeclarelayer{foreground}
\pgfsetlayers{background,main,foreground}

%% file: root.bbl
\begin{thebibliography}{10}

\bibitem{AAM1}
R.~Goyal, C.~Reiche, C.~Fernando, and A.~Cohen, ``Advanced air mobility: Demand
  analysis and market potential of the airport shuttle and air taxi markets,''
  {\em Sustainability}, vol.~13, no.~13, 2021.

\bibitem{AAMFund1}
C.~Dietrich, T.~Johnston, and R.~Riedel, ``{Looking to the skies: Funding for
  future air mobility takes off},'' {\em McKinsey {\&} Company}, 2021.

\bibitem{NASAAAM}
``{Regional Air Mobility: Leveraging our National Investments to Energize the
  American Travel Experience},'' tech. rep., NASA Langley Research Center,
  2021.

\bibitem{giribet2016analysis}
J.~I. Giribet, R.~S. Sanchez-Pe\~na, and A.~S. Ghersin, ``Analysis and design
  of a tilted rotor hexacopter for fault tolerance,'' {\em IEEE Transactions on
  Aerospace and Electronic Systems}, vol.~52, no.~4, pp.~1555--1567, 2016.

\bibitem{Michieletto2017}
G.~{Michieletto}, M.~{Ryll}, and A.~{Franchi}, ``Control of statically
  hoverable multi-rotor aerial vehicles and application to rotor-failure
  robustness for hexarotors,'' in {\em 2017 IEEE International Conference on
  Robotics and Automation (ICRA)}, pp.~2747--2752, May 2017.

\bibitem{TMECH2020}
C.~Pose, J.~I. Giribet, and I.~Mas, ``{Fault Tolerance Analysis for a Class of
  Reconfigurable Aerial Hexarotor Vehicles},'' {\em IEEE/ASME Transactions on
  Mechatronics}, vol.~4, p.~1851–1858, 8 2020.

\bibitem{Abbaraju2021}
P.~Abbaraju, X.~Ma, G.~Jiang, M.~Rastgaar, and R.~M. Voyles, ``Aerodynamic
  modeling of fully-actuated multirotor uavs with nonparallel actuators,'' in
  {\em 2021 IEEE/RSJ International Conference on Intelligent Robots and Systems
  (IROS)}, pp.~9639--9645, 2021.

\bibitem{rasmussen2006gaussian}
C.~E. Rasmussen and C.~K. Williams, {\em {Gaussian} processes for machine
  learning}, vol.~1.
\newblock MIT press Cambridge, 2006.

\bibitem{beckers2021introduction}
T.~Beckers, ``An introduction to gaussian process models,'' {\em arXiv preprint
  arXiv:2102.05497}, 2021.

\bibitem{hewing2019cautious}
L.~Hewing, J.~Kabzan, and M.~N. Zeilinger, ``Cautious model predictive control
  using gaussian process regression,'' {\em IEEE Transactions on Control
  Systems Technology}, vol.~28, no.~6, pp.~2736--2743, 2019.

\bibitem{lima2020sliding}
G.~S. Lima, S.~Trimpe, and W.~M. Bessa, ``Sliding mode control with gaussian
  process regression for underwater robots,'' {\em Journal of Intelligent \&
  Robotic Systems}, vol.~99, no.~3, pp.~487--498, 2020.

\bibitem{beckers2019stable}
T.~Beckers, D.~Kuli{\'c}, and S.~Hirche, ``Stable {G}aussian process based
  tracking control of {E}uler--{L}agrange systems,'' {\em Automatica},
  vol.~103, pp.~390--397, 2019.

\bibitem{beckers2021online}
T.~Beckers, L.~J. Colombo, S.~Hirche, and G.~J. Pappas, ``Online learning-based
  trajectory tracking for underactuated vehicles with uncertain dynamics,''
  {\em IEEE Control Systems Letters}, vol.~6, pp.~2090--2095, 2022.

\bibitem{isidori1985nonlinear}
A.~Isidori, {\em Nonlinear control systems: an introduction}.
\newblock Springer, 1985.

\bibitem{aastrom1971system}
K.~J. {\AA}str{\"o}m and P.~Eykhoff, ``System identification—a survey,'' {\em
  Automatica}, vol.~7, no.~2, pp.~123--162, 1971.

\bibitem{wolpert1996lack}
D.~H. Wolpert, ``The lack of a priori distinctions between learning
  algorithms,'' {\em Neural computation}, vol.~8, no.~7, pp.~1341--1390, 1996.

\bibitem{wahba1990spline}
G.~Wahba, {\em Spline models for observational data}.
\newblock SIAM, 1990.

\bibitem{steinwart2008support}
I.~Steinwart and A.~Christmann, {\em Support vector machines}.
\newblock Springer Science \& Business Media, 2008.

\bibitem{srinivas2012information}
N.~Srinivas, A.~Krause, S.~M. Kakade, and M.~W. Seeger, ``Information-theoretic
  regret bounds for {Gaussian} process optimization in the bandit setting,''
  {\em IEEE Transactions on Information Theory}, vol.~58, no.~5,
  pp.~3250--3265, 2012.

\bibitem{Berkenkamp2016ROA}
F.~Berkenkamp, R.~Moriconi, A.~P. Schoellig, and A.~Krause, ``Safe learning of
  regions of attraction for uncertain, nonlinear systems with {G}aussian
  processes,'' in {\em Proc. of the IEEE Conference on Decision and Control},
  pp.~4661--4666, 2016.

\bibitem{liberzon1999basic}
D.~Liberzon and A.~S. Morse, ``Basic problems in stability and design of
  switched systems,'' {\em IEEE Control Systems Magazine}, vol.~19, no.~5,
  pp.~59--70, 1999.

\bibitem{lee2012exponential}
T.~Lee, ``Exponential stability of an attitude tracking control system on so
  (3) for large-angle rotational maneuvers,'' {\em Systems \& Control Letters},
  vol.~61, no.~1, pp.~231--237, 2012.

\bibitem{UnmSyst2020}
C.~Pose, F.~Presenza, I.~Mas, and J.~I. Giribet, ``Trajectory following with a
  {MAV} under rotor fault conditions,'' {\em Unmanned Systems}, vol.~8, no.~4,
  pp.~263--268, 2020.

\bibitem{VideoPathFollow}
``Path following with a fault-tolerant hexarotor.''
  \url{https://youtu.be/MR_4ccN5ECg}.
\newblock Accessed: 2022-02-21.

\bibitem{VideoGPbased}
``Learning-based fault-tolerant control for anhexarotor with model
  uncertainty.'' \url{https://youtu.be/tztkdIxMM2Y}.
\newblock Accessed: 2022-02-21.

\end{thebibliography}
